\documentclass[preprint]{elsarticle}
\usepackage[utf8]{inputenc}
\usepackage[american]{babel}
\usepackage{amsfonts, amsmath, amssymb}
\usepackage{mathtools}
\usepackage{fancyhdr}
\usepackage{amsthm}
\usepackage{graphicx}
\usepackage{changepage}
\usepackage{algorithm}
\usepackage{algpseudocode}
\usepackage{tikz}
\usepackage{subfig}
\usepackage{doi}

\newtheorem{thm}{Theorem}  \newdefinition{rmk}{Remark} \newproof{pf}{Proof}
\newproof{pot}{Proof of Theorem \ref{thm2}}

\title{Determining Implication of Fixed Matrix Prenex Normal Forms Can Be Decided in Linear Time}
 \author[1]{Adam Wang\corref{cor1}}
\cortext[cor1]{Corresponding author}
 \affiliation[1]{
 organization={Decamicron Labs},
city={Baltimore},
postcode={21218},
state={Maryland},
country={USA}}

\begin{document}

\begin{abstract}
    For a fixed arbitrary matrix depending on $n$ variables, one may ask whether a Prenex Normal Form (PNF) implies another. A RAM algorithm running in linear time is presented and shown to be asymptotically optimal.
\end{abstract}

\maketitle

\section{Introduction}
 Any given logical sentence may be brought into Prenex Normal Form (PNF) \cite{hinman_fundamentals_2018}. This consists of a string of $n$ quantified variables, known as the prefix, as well as a string of propositions, known as the matrix. 

 If one restricts the variables to be boolean, one recovers the Quantified Boolean Formulas (QBF). The simpler case where all the quantifiers are existential is the famous Boolean satisfiability problem (SAT), which is known to be NP-Complete \cite{karp_reducibility_1972}. The general case for all quantifiers is PSPACE-Complete, and so is the restriction to the bounded-width case \cite{atserias_bounded-width_2014}.

 In the general setting, various results have focused upon fixed forms of quantifiers for a function-symbol-free matrix. For example, the Ackermann, Gödel, and Bernays–Schönfinkel classes have been studied. It is known that the latter is NEXPTIME-complete \cite{lewis_complexity_1980}. 

 Despite the difficulty of these problems, certain restrictions lead to much simpler problems. For example, if one restricts SAT such that sentences are conjunctions of clauses, where each clause is a disjunction of two variables, then a famous result tells us this may be solved in linear time \cite{aspvall_linear-time_1979}\cite{krom_decision_1967}. 

 If one restricts to implications, another well-studied class of problems involves the Horn clauses, which are disjunctions with at most one non-negated term. One can check these clauses always correspond to an implication statement. The boolean satisfiability of Horn clauses (HORNSAT) is known to be P-Complete \cite{cook_logical_2010}. In fact, it is linear \cite{dowling_linear-time_1984}, but becomes polynomial in the multi-valued logic case \cite{beckert_transformations_1999}\cite{hahnle_tutorial_2001}. A quantified boolean case has also been considered and shown to run in linear time \cite{buning_resolution_1995}. 

The present paper considers a different direction in the vein of quantified implications. Instead of boolean logic, this problem falls under the general setting. For two sentences $S_1$ and $S_2$ in PNF with the same arbitrary matrix, does $S_1$ imply $S_2$? This problem will be termed PNF Prefix Implication (PNFPI). More explicitly, for $n$ variables $\textbf{x}$ that each have non-empty range and an arbitrary PNF matrix $P(\textbf{x})$, if $Q_1\textbf{x}$ and $Q_2\textbf{x}$ are valid PNF prefixes, then when do we always have that:
\begin{align*}
    Q_1\textbf{x} P(\textbf{x}) \implies Q_2\textbf{x} P(\textbf{x}) 
\end{align*}
Or as a quantified Horn clause:
\begin{align*}
    (\neg Q_1)\textbf{x}_1Q_2\textbf{x}_2 (\neg P(\textbf{x}_1)\vee P(\textbf{x}_2))
\end{align*}

The main result of this paper will be to show that there exists a Random Access Machine (RAM) algorithm linear in prefix length that solves PNFPI. 

\section{The Problem}
Begin by noting that since the matrix is arbitrary and the same, only the prefix is salient. Therefore, PNFPI takes in two pieces of information: firstly, an ordered list $\sigma$ of numbers from $0$ to $n-1$, and secondly, a Boolean string $B$ of length $n$. The first list specifies the order of the variables and the latter the quantifiers at each position. 

Henceforth, in the context of the Boolean string $B$, let $0$ represent $\exists$ and $1$ represent $\forall$. A \emph{$\forall$-run} is the locally longest sub-string of contiguous variables with the universal quantifier. A \emph{$\forall$-element} is an element of a $\forall$-run. Similarly define a \emph{$\exists$-run} and a \emph{$\exists$-element}. Elements are $0$-indexed from the left.

Note that two prefixes are logically equivalent up to partial permutation within each run of quantifiers of the same type. For example, the following two prefixes are equal since one permutes $x_1$ and $x_2$ within the same $\forall$-run:
$$\forall x_1 \forall x_2 \exists x_3 \forall x_4 \quad\quad\forall x_2 \forall x_1 \exists x_3 \forall x_4$$
However:
$$\forall x_1 \forall x_4 \exists x_3 \forall x_2 $$
is not equal to the former two sentences for an arbitrary matrix.

Next, note that it was assumed that each variable has a non-empty range. Thus, if all elements of a range satisfy a matrix, then there must exist at least one element of that range that satisfies it. Further, if there exists an element in range A such that for all elements of range B a matrix holds, then for all elements of range B there exists an element in range A where the matrix holds. However, in both cases, the reverse direction is not always true for an arbitrary matrix. These are also the only two implications that apply universally. If there is an implication between two PNFs of the same matrix, then there must be a sequence of the above two moves from the first equivalence class to the latter.

Since equivalence classes are hard to deal with computationally, representatives are considered instead. Given any representative of an equivalence class of logical sentences, the previous considerations lead to three possible moves:
\begin{enumerate}
    \item $\forall x_1\forall x_2\implies\forall x_2\forall x_1$ or $\exists x_1\exists x_2 \implies \exists x_2 \exists x_1$
    \item $\forall x \implies\exists x$
    \item $\exists x_1 \forall x_2\implies \forall x_2 \exists x_1$
\end{enumerate}

We may recast this problem in graph-theoretic language to get a sense of the size of the search space. One may form a directed graph of all equivalence classes in $n$ variables and their implications. If $e(n)$ is the number of edges in this graph, then a naive tree search has a time complexity of at least \cite{oeis_foundation_inc_entry_2025}:
$$\mathcal O \left(\frac{n!}{\log(2)^{n+1}}+e(n)\right)$$
The directed graph of equivalence class representatives grows even faster. However, it is still finite, which proves that PNFPI is always solvable for a finite input.

\section{The Algorithm}
\begin{algorithm}
\caption{RAM Linear time algorithm for PNFPI}\label{alg:cap}
\begin{algorithmic}
\Require $S_1 = \{\sigma_1, B_1\}, S_2= \{\sigma_2, B_2\}, n=|\sigma_1|=|\sigma_2|=|B_1|=|B_2|$ \Comment{Two prefixes of same length}
\Ensure \textbf{accept} if $(S_1\implies S_2)$, \textbf{reject} otherwise 

\State $P \gets$ hash-map of elements to position in $\sigma_1$
\State $F \gets$ index of largest indexed $\exists$-element in $S_1$
\State $V \gets$ length $n$ list of $0$s

\State $i \gets n-1$
\While{$i \geq 0$}
    \If{$B_1[P[\sigma_2[i]]]$ == $\exists$ and $B_2[i]$ == $\forall$}
        \State \Return \textbf{reject} \Comment{Case 5}
    \ElsIf{$B_1[P[\sigma_2[i]]]$ == $\forall$ and $B_2[i]$ == $\forall$}
        \If{$F > P[\sigma_2[i]]$}
            \State \Return \textbf{reject} \Comment{Case 4}
        \EndIf
    \EndIf
    \State $V[P[\sigma_2[i]]]\gets 1$ \Comment{Record $\sigma_2[i]$ has been verified}
    \If{$P[\sigma_2[i]]==F$}
        \State $F \gets$ index of largest indexed unverified $\exists$-element in $S_1$, computed using $V$
    \EndIf
    \State $i\gets i-1$
\EndWhile
\State \Return \textbf{accept}
\end{algorithmic}
\end{algorithm}
\begin{thm}
    Algorithm 1 correctly solves PNFPI
\end{thm}
\begin{proof}
    Since the algorithm always halts and returns for a finite input, it suffices to prove that the algorithm accepts if and only if $S_1\implies S_2$.

    \emph{Forward direction.} Assume the algorithm returns $\textbf{accept}$. It will be claimed that this implies we can construct a chain of implications between $S_1$ and $S_2$. 

    Firstly, consider the following reduced problem. Given two prefixes $I_k$ and $S_2$ that agree in the last $k$ elements, the goal is to construct $I_{k+1}$ such that $I_k \implies I_{k+1}$, and that $I_{k+1}$ agrees with $S_2$ in the last $k+1$ elements. If this is always possible, then we can construct the aforementioned chain. 

    Note that the $(n-k-1)$-th position is the largest indexed position of $I_k$ that is not guaranteed to agree with $S_2$. If $x_k$ is the $(n-k-1)$-th element in $S_2$, then we attempt to move $x_k$ in $I_k$ to the $(n-k-1)$-th position while matching the quantifiers.

    There are a total of $5$ cases:
    \begin{itemize}
    \item \textbf{Case 1} ($x_k$ is an $\exists$-element in both $I_k$ and $S_2$) From the first and third allowed moves, one can always move an $\exists$-element to any larger indexed position. Since $I_k$ and $S_2$ agree for the last $k$ indices and variables are unique, then the index of $x_k$ must be strictly less than $n-k$.
    
    \item \textbf{Case 2} ($x_k$ is an $\forall$-element in $I_k$ and $\exists$-element in $S_2$) We may simply note that we can apply the second move to switch quantifiers and this reduces to Case 1. 

    \item \textbf{Case 3} ($x_k$ is an $\forall$-element in both $I_k$ and $S_2$, and that $x_k$ is in the last run) If there are no runs behind the run $x_k$ is in, then any elements with larger index than $x_k$ in $S_1$ are $\forall$-elements. Therefore we can move $x_k$ to the correct position using the first move. 
        
    \item \textbf{Case 4} ($x_k$ is an $\forall$-element in both $I_k$ and $S_2$, and that $x_k$ is not in the last run) If $x_k$ is not in the last run, there must be an $\exists$-run after. There is no clear way to construct $I_{k+1}$.
    
    \item \textbf{Case 5} ($x_k$ is an $\exists$-element in $I_k$ and $\forall$-element in $S_2$) There is no clear way to construct $I_{k+1}$.
    \end{itemize}

    Thus, if one encounters Cases 1, 2, and 3, then it is possible to construct $I_{k+1}$. If one encounters Case 4 and 5, then it is not clear if we can. It will later be shown that this is indeed impossible, but this is not needed for the forward direction. 

    Algorithm 1 checks whether it is possible to build each $I_k$ interpolating between $S_1$ and $S_2$. It goes through each element from the back of $S_2$, checks for Case 4 and 5, then virtually builds each $I_k$. To do so, it suffices to keep track of the index $F$ of the largest indexed unverified $\exists$-element in $S_1$. 
    
    It first checks Case 5, which is done by straightforward matching of quantifiers. 
    
    Next, it checks whether we have Case 3 or 4. This amounts to first determining whether the quantifiers are both universal, and then comparing the index of $x_k$ to $F$. If $F$ is larger, then there is an $\exists$-element behind $x_k$. As $x_k$ is an $\forall$-element, it cannot be in the last run, which is Case 4. If not, then the last $\exists$-run is before $x_k$. Thus, any elements with larger index are $\forall$-elements, so we have Case 3. If it encounters Case 4 or 5, it returns $\textbf{reject}$. Else, it is possible to construct $I_{k+1}$, so it updates $F$, thereby regenerating the necessary initial conditions for the next iteration and continues. 

    Therefore, if the algorithm returns $\textbf{accept}$, then only Case 1, 2, and 3 were encountered so it is possible to construct a chain of implications from $S_1$ to $S_2$. 

    \emph{Backward direction.} Assume that $S_1\implies S_2$. We proceed by contradiction. As the input is finite, the algorithm must return, so it suffices to assume that the algorithm returns $\textbf{reject}$. This implies that either Case 4 or Case 5 was encountered at some point. 

    \begin{itemize}
    \item \textbf{Case 4} ($x_k$ is an $\forall$-element in both $I_k$ and $S_2$, and that $x_k$ is not in the last run) For some $I_k$, there is a $\exists$-run behind a $\forall$-run. Let $x_1$ denote the $\forall$-element and $x_2$ denote the $\exists$-element. In $S_1$, the index of $x_1$ must be smaller than the index of $x_2$ since the construction of each $I_k$ keeps the relative order of the unverified elements unchanged. Meanwhile in $S_2$ the index of $x_1$ is larger than the index of $x_2$, because both must be unverified based on their positions in $I_k$, and $x_1$ must have the largest index of all unverified elements in $S_2$. If $S_1\implies S_2$ then there must be a sequence of implications using the three allowed moves that switches the index of $x_1$ to be larger than $x_2$. However, it is impossible to move a $\forall$-element behind a $\exists$-element while retaining the same quantifier type.
    
    \item \textbf{Case 5} ($x_k$ is an $\exists$-element in $I_k$ and $\forall$-element in $S_2$) There must be a variable $x_k$ that is a $\exists$-element in $S_1$ and $\forall$-element in $S_2$. However, it is impossible to switch a $\exists$-element to an $\forall$-element. 
    \end{itemize}
    If either case was encountered, $S_1$ cannot imply $S_2$.
\end{proof}
\begin{thm}
    Algorithm 1 runs in RAM $\mathcal O(n)$ time 
\end{thm}
\begin{proof}
    Firstly, the hashmap of elements $P$ may be constructed in linear time by initializing it as an empty list of length $n$, then iterating over $i$ and setting the $\sigma_1[i]$-th element of $P$ to $i$. 

    Secondly, the position of the final $\exists$-element in $S_1$ may be found in linear time by iterating from the back of $B_1$. 

    Thirdly, the \texttt{while} loop will have at most $n$ stages so it remains to show that each stage runs in constant time with an overhead that adds up to at most linear. Everything but the final \texttt{if} block will run in constant time, because the only operations are accessing arrays and performing logical inference. 
    
    One may achieve linear overhead for the final \texttt{if} block by beginning at the $F$-th element of $V$ and $B_1$, and iterating downwards until the first occurrence of an unverified $\exists$-element. Since we begin at $F$ each time, over the course of the entire \texttt{while} loop, this process can only iterate at most once through the two lists, so the total overhead is at most linear.
\end{proof}

\begin{thm}
    Any RAM algorithm that solves PNFPI is at least linear time
\end{thm}
\begin{proof}
    Consider the reduced problem where the input is two prefixes $S_1$ and $S_2$ such that: 
    $$\sigma_1=\sigma_2$$ 
    Since in our boolean encoding, $\forall > \exists$, the problem reduces to checking that $B_1$ is element-wise larger than $B_2$. This requires at least $n$ accesses from $B_1$. Suppose some element was not accessed, then the output of the algorithm cannot depend on this element. However, if that element were $\forall$ in $B_2$, its value in $B_1$ would give different answers. 
\end{proof}

The algorithm allows one to compute the probability of implication between two uniformly sampled sentences in linear memory. Equivalently, this is the total number of true or false implications given $n$ variables. This is a lower bound on the actual number of implications for a specific matrix. However, the time complexity of this naive approach is poor and only permits feasible computation for very small $n$. Does there exist a polynomial time and memory algorithm to compute this quantity? What about to efficiently generate all sufficient or necessary statements given some statement? 

Many symbolic logic solvers currently exist in the market, some with proprietary algorithms. Due to the restrictive nature of this algorithm, it is unlikely to find application beyond as a leetcode question. 

\bibliographystyle{plainnat} 
\bibliography{main} 
\end{document}